\RequirePackage{fix-cm}
\documentclass[smallextended]{svjour3}       
\smartqed  
\usepackage{amssymb,latexsym,amsfonts,amsmath,mathrsfs}
\usepackage{graphicx}
\usepackage{geometry}
\usepackage{dsfont}
\usepackage{setspace}

\usepackage{hyperref}
\usepackage{subfigure}

\hypersetup{colorlinks=true, linkcolor=blue, citecolor=red}

\allowdisplaybreaks

\newcommand{\xqedhere}[2]{%
  \rlap{\hbox to#1{\hfil\llap{\ensuremath{#2}}}}}

\begin{document}

\title{Global Properties of Nested Network Model with Application to Multi-Epitope HIV/CTL Dynamics}


\author{Cameron Browne         
}


\institute{Mathematics Department, University of Louisiana at Lafayette, Lafayette, LA \\
            \email{cambrowne@louisiana.edu}           
  }

\maketitle

\begin{abstract}
Mathematical modeling and analysis can provide insight on the dynamics of ecosystems which maintain biodiversity in the face of competitive and prey-predator interactions.   Of primary interests are the underlying structure and features which stabilize diverse ecological networks.  Recently Korytowski and Smith \cite{korytowski2015nested} proved that a perfectly nested infection network, along with appropriate life history trade-offs, leads to coexistence and persistence of bacteria-phage communities in a chemostat model.  In this article, we generalize their model in order to apply it to the within-host dynamics virus and immune response, in particular HIV and CTL (Cytotoxic T Lymphocyte) cells.  Our model can describe sequential viral escape from dominant immune responses and rise in subdominant immune responses, consistent with observed patterns of HIV/CTL evolution.  We find a Lyapunov function for the system which leads to rigorous characterization of persistent viral and immune variants, along with informing upon equilibria stability and global dynamics.  Results are interpreted in the context of within-host HIV/CTL evolution and numerical simulations are provided. 
\end{abstract}

\section{Introduction}\label{s2}

Mathematical models of ecological communities with both competitive and predator-prey interactions have often been studied in order to determine structure, persistence and dynamics of ecosystems.   The general model setup can include $n$ predators which predate on $m$ competing prey species, where an $m\times n$ matrix describes the bipartite network of interactions between prey and predators.  In this paper, we focus on a specific structure for the interaction network, namely a perfectly nested network.  In a perfectly nested scenario, the specialist predator can target the most susceptible prey, the next most specialized predator
targets the most susceptible prey and the second most susceptible prey, and so on.  Nested networks have been of recent interest in explaining the biodiversity and structure often observed in bacteria phage communities \cite{jover2013mechanisms,korytowski2015nested,korytowski2,weitz2013phage}.  Jover et al. show that the hierarchical nested structure can be maintained in a Lotka-Volterra model with trade-offs between the competitive ability of
the bacteria hosts and defense against infection and, on
the part of virus, between virulence and transmissibility
versus host range \cite{jover2013mechanisms}.  Korytowski and Smith consider a chemostat model with a perfectly nested bacteria phage network and provide a rigorous argument for persistence of the community when these trade-offs are assumed \cite{korytowski2015nested}.  In this article we extend the chemostat model of Korytowski and Smith \cite{korytowski2015nested} by including mortality of the prey and predators (bacteria and phages, respectively, in their context).  

Our primary biological motivation is to apply the model to the within-host population dynamics of virus and immune response (interpreted as prey and predator, respectively, here).  In particular, we consider interactions between HIV and CTL (Cytotoxic T Lymphocyte) immune effectors.  CTL immune effectors recognize pathogen-derived proteins (epitopes) presented on the surface of infected cells to mediate their killing.  There is an extensive repertoire of distinct CTL clones targeting different epitopes.  However, the ability of HIV to rapidly evolve allows for the rise of several mutant strains with mutations conferring resistance to attack at different epitopes.  The ensuing arms race creates an evolving network of viral strains and CTL with variable levels of \emph{(strain) reactivity}, i.e. epitopes shared between virus strains.  

There is extensive research on modeling the within-host dynamics of multiple variants of virus and immune response assuming different virus-immune network structures.  One class of HIV/CTL multi-variant models, first studied by Nowak et al. \cite{nowak1996population}, are extensions of the standard virus model \cite{Perelson2} and consider strain-specific CTLs (one-to-one reactivity) with mass-action killing and activation rates.  In this case, the diversity of the persistent virus set has been characterized in terms of model parameters \cite{bobko2015singularly,iwasa2004some,souza2011global}.   However the assumption of strain-specific immune response does not correspond to the biological reality that CTLs are specific to epitopes and, in general, multiple epitopes will be shared among virus strains.   Another set of models \cite{Nowak2} includes growth of multiple virus strains (without the target cell limitation present in the aforementioned standard virus model) inhibited by CTL immune responses directed at two distinct epitopes.  In this multi-epitope model, the antigenic variation induces oscillations (antigenic oscillations) and coexistence of the distinct CTLs \cite{Nowak2}.  Similar models (with only linear immune activation rates, as opposed to mass-action) have been studied in the context of malaria infection and also have displayed antigenic oscillations, along with complex dynamics \cite{Gupta,PilyuginMalaria,Blyuss}.  

While the virus-immune epitope interaction network generally can be quite complex, patterns of viral escape, immunodominance hierarchies and cross-reactivity often emerge.  A nested interaction network can be representative of successive (partial) escape from a dominant epitope followed by the rise of a subdominant epitope, resulting in a nested collection of viral strains where the specialist dominant immune response attacks only the wild-type virus strain and the weaker immune response attacks a broad range of viral strains with varying resistance profiles (attacking at a conserved epitope).  This constrains the mutational pathway of the virus so that resistance to multiple epitopes is built sequentially in a perfectly nested order.  The successive rise of more broadly resistant prey (coming with a fitness cost) and weaker but more generalist predators, in a perfectly nested fashion, is the route to persistence of nested bacteria-phage communities argued in \cite{korytowski2015nested}.

There has also been much interest in quantifying rates, factors and patterns of viral immune escape from multiple epitopes \cite{Althaus,Vitaly1,Vitaly2,liu2013vertical,kessinger2015inferring,leviyang2015broad,pandit2014reliable,vanDeutekom}.  Models with the full range of immune epitope reactivity networks have been considered, but analysis of such models is limited in general by the multitude of possible escape pathways which may be explored concurrently (for $n$ epitopes there are $2^n$ possible paths to full resistance).  The special case of nested emergence of multi-epitope resistance has been studied and there is some evidence that nestedness is a feature of HIV-CTL interactions \cite{kessinger2015inferring,liu2013vertical,vanDeutekom}.  Indeed, the immune escapes in the early phase of HIV infection are shown to occur sequentially due to strong selection \cite{da2012dynamics,kessinger2015inferring}.  Although concurrent escape through different pathways does occur during chronic infection, it is often resolved in favor of a given escape variant, which then serves as the basis of another concurrent escape that is resolved in favor of another escape variant, and so on; resulting in a highly nested network of multi-epitope resistance \cite{leviyang2015broad,pandit2014reliable,vanDeutekom}.  Also, we will show that successive viral invasion and rise of subdominant immune response becomes slower and rarer in a nested network as it's size (breadth of virus/immune diversity) increases, which is consistent with observed patterns in HIV.   In addition, a perfectly nested network offers a particular case of multi-epitope interactions which is amenable to mathematical analysis and can promote coexistence, as shown for phage-bacteria models in \cite{korytowski2015nested}.  Analysis of more general virus-immune networks, which allow for multiple escape pathways, will be considered in a future article.

Ecological networks of prey and predators, such as the virus-immune response and bacteria phage examples, can exhibit significant biodiversity and nestedness offers a compelling way to explain the coexistence of multiple competing species.  Mathematically the question of coexistence versus competitive exclusion of species can be framed in terms of uniform persistence or permanence \cite{hofbauer1998evolutionary,smith2011dynamical}.  Korytowski and Smith prove uniform persistence of the nested bacteria phage community in a chemostat model \cite{korytowski2015nested}.  In this paper, we provide criteria for uniform persistence or extinction of the populations in our generalization
of their chemostat model.  The method of proof utilized here is different than in \cite{korytowski2015nested}, in particular we find a Lyapunov function, which allows for a more complete global analysis.  In particular, we rigorously characterize which species persist and which die out.  While we are not able to prove global convergence to equilibria in general, global stability is established in the case of two or less persistent immune responses.  Also our method highly constrains the dynamics on the attracting set, so that we conjecture global stability of the appropriate equilibrium to generally hold.  Our Lyapunov function approach is similar to works by Bobko and Zubelli \cite{bobko2015singularly} on the one-to-one virus-immune reactivity network model, Korytowski and Smith \cite{korytowski2} on one-to-one and nested phage-bacteria networks in special cases of the Lotka-Volterra model, and Wolkowicz \cite{wolkowicz1989successful} on food webs in a chemostat model.

This article is organized as follows.  In Section \ref{SecModel}, the general mathematical model of HIV/CTL dynamics is introduced and shown to be well-posed, and assumptions along with motivation are provided for the case of perfectly nested interaction networks.  In Section \ref{SecDynamics}, the main result (Theorem \ref{mainThm}) is stated, which gives conditions for the persistence of distinct viral and immune variants in the model and describes the ensuing dynamics with respect to system equilibria.  In Section \ref{SecApp}, analytical results are discussed in the context of within-host HIV/CTL evolution, accompanied by numerical simulations.  Finally, Section \ref{SecProofs} contains proofs of the theorems, in particular the main result is proved via Lyapunov functions and applications of LaSalle's invariance principle.

\section{Mathematical model} \label{SecModel}
\subsection{General model \& boundedness of solutions}

As mentioned in the introduction, our primary motivation is the within-host dynamics of HIV and CTL cells, thus we view the prey and predators in our model as virus and immune responders respectively.  Consider the following differential equation system:
\begin{align}
\frac{dX}{dt} &= b-aX- X\sum_{i=1}^n \beta_iY_i, \notag \\
\frac{dY_i}{dt} &= \beta_iY_iX-\delta_iY_i-Y_i\sum_{j=1}^m r_{ij} Z_j , \quad i=1,\dots,m \label{ode2}  \\
  \frac{dZ_j}{dt} &= Z_j\sum_{i=1}^n q_{ij}Y_i-  \mu_jZ_j, \notag \quad j=1,\dots,n.
\end{align}

The variables $X(t)$ and $Y_i(t)$ denote the concentration of uninfected target cells, and infected cells of virus strain $i$, respectively.    $Z_j(t)$ is the concentration of immune effector cells responding to a particular epitope, labeled $j$.   The function $f(X)=b-aX$ represents the net growth rate of the uninfected cell population.  The parameters $\beta_i$ and $d_i$ are the infection rate and $\delta_i$ is the death rate for infected cells infected with virus strain $i$.  The parameter $r_{ij}$ describes the killing rate of immune population $Z_j$ on virus strain $i$, whereas $q_{ij}$ describes the corresponding activation rate for $Z_j$.   The parameter $\mu_j$ denotes the death rate of the immune effector cells.  In the present paper, we assume that virus load (the abundance of virions) is proportional to the amount of (productively) infected cell.  This assumption has frequently been made for HIV since the dynamic of free virions occurs on a much faster time scale than the other variables.  Another reason for not explicitly including free virus in our present work is to keep the tri-trophic ecological structure of resource-prey-predators, allowing for more general applicability in community ecology.

Note that the mass-action forms of the immune killing and activation rates are representative of actions that occur proportionally to the strength of interaction, denote by $\psi_{ij}$, between immune cells and epitopes on the surface of infected cells.  Thus it is reasonable to make the assumption that the killing and activation rate for $Z_j$ are proportional to $\psi_{ij}$, i.e. of the form $r_{ij}=r_j\psi_{ij}$ and $q_{ij}=q_jr_j\psi_{ij}$, respectively.  The $m\times n$ matrix $R=\left(r_{ij}\right)$ gives the reactivity of the epitopes with respect to the antigenic (virus) variation.  Let $\mathcal S_j:=\left\{ i\in [1,m] : r_{ij}>0 \right\}$ denote the reactivity range for each immune response $Z_j$.    
The model becomes:
\begin{align}
\frac{dX}{dt} &= b-aX- X\sum_{i=1}^n \beta_iY_i, \notag \\
\frac{dY_i}{dt} &= \beta_iY_iX-\delta_iY_i-Y_i\sum_{j: i\in \mathcal S_j} r_{ij} Z_j , \quad i=1,\dots,m \label{ode3} \\
  \frac{dZ_j}{dt} &= q_jZ_j\sum_{i\in \mathcal S_j} r_{ij}Y_i-  \mu_jZ_j, \quad j=1,\dots,n. \notag
\end{align}

\begin{proposition}\label{bounded}
Consider the system (\ref{ode3}) with non-negative initial conditions in $\mathbb R^{m+n+1}$.  Solutions remain non-negative for all time $t$ and there exists a bounded set in $\mathbb R^{m+n+1}$ which attracts all solutions.  
\end{proposition}

\subsection{Perfectly nested network} \label{NestedIntro}
In the remainder of this paper, we assume a perfectly nested interaction network structure.  Explicitly, suppose that $m=n$ and for each immune response $Z_j$, $j=1,\dots n$, the reactivity range is $\mathcal S_j=[1,j]$.  Then the reactivity matrix $R$ is $n\times n$ and upper triangular.  We also make the following assumptions: 
\begin{itemize}
\item[(i)] $\forall j, \  i\in \mathcal S_j \Rightarrow r_{ij}=r_j>0$, i.e. immune response killing/interaction rate is independent of virus strain in reactivity range.    \item[(ii)] $\delta_i=\delta$ for $i=1,\dots,n$.   
\end{itemize}
We now rescale the model and introduce the following quantities: 
\begin{align*}
x&=\frac{a}{b}X, \quad y_i=\frac{\delta}{b}Y_i, \quad z_j=\frac{r_j}{ \delta}Z_j, \\
\tau&=at, \quad \gamma=\frac{\delta}{a}, \quad \sigma_j=\frac{\mu_j}{ a}, \\
\mathcal R_i&=\frac{b\beta_i}{a\delta},  \quad \mathcal I_j=\frac{bq_jr_j}{ \delta\mu_j}, 
\end{align*}
where $\mathcal R_i$ represents the basic reproduction number of virus strain $i$ and $\mathcal I_j$ represents the basic reproduction number of immune population $z_j$.  For convenience, we introduce the following notation:
\begin{align*}
s_i=\frac{1}{\mathcal I_i}, \qquad  s_0=0
\end{align*}
Denoting the derivative $\frac{d}{d\tau}$ with ``dot notation'', the model becomes:
\begin{align}
\dot x &= 1-x- x\sum_{i=1}^n \mathcal R_i y_i, \notag \\
\dot y_i &= \gamma y_i\left(\mathcal R_i x -1 -\sum_{j\geq i} z_j \right), \quad i=1,\dots,n \label{ode4}  \\
  \dot z_i &= \frac{\sigma_i}{s_i} z_i\left(\sum_{j\leq i} y_j -s_i \right), \quad i=1,\dots,n \notag
\end{align}
Note the difference between our above model (\ref{ode4}) and the chemostat model analyzed by Korytowski and Smith  \cite{korytowski2015nested} is that the latter assumes a single dilution rate with no mortality allowing for the dimension of the system to be reduced, which was crucial for their method of analysis.  

The perfectly nested structure in model (\ref{ode4}) can arise when considering an immunodominance hierarchy and sequential viral escape from successive immune responses.   Suppose there are $n$ epitopes each targeted by a distinct CTL line $z_1,\dots,z_n$.  Generally, there are $2^n$ possible viral strains based on whether they are resistant to attack on a particular epitope, however if mutation follows a perfectly nested pattern the dimension is reduced substantially.  Without loss of generality, order the immune responses in decreasing order with respect to CTL reproduction number: 
$ \mathcal I_1>\mathcal I_2>\dots > \mathcal I_n$.
Initially an individual is infected with a single virus strain, called the wild type or founder strain, denoted here by the variable $y_1$.  The dominant CTL immune response $z_1$ will attack this strain imposing selection pressure for an escape mutant.  In general, $y_1$ is subject to attack at $n$ epitopes, and both immunodominance and the fitness costs of mutation may factor into the escape pathway followed.    However in building the perfectly nested network, we assume that the first mutant virus, $y_2$, is resistant to $z_1$, with a fitness cost associated to the mutation.   Indeed we might argue that if mutation is sufficiently rare, the most dominant immune response competitively excludes all others with only one virus strain present, so that an initial mutant $y_2$ only would arise resistant to $z_1$.   After the mutant $y_2$ rises, the subdominant immune response $z_2$ against another epitope will rise.   Although it is possible that a mutant will rise which is resistant to $z_2$ and susceptible to $z_1$ \cite{batorsky2014route}, we assume that the next mutant will be resistant to both $z_2$ and $z_1$ as in \cite{kessinger2015inferring}.  This nested sequential pattern of multi-epitope resistance may continue in this fashion, with $z_n$ being the ``broadest and weakest'' immune response, and $y_n$ or $y_{n+1}$ having accumulated mutations at $n-1$ or $n$ epitopes, each with a fitness cost.

\section{Dynamics of Perfectly Nested Model} \label{SecDynamics}
  Consider the rescaled model with perfectly nested interaction network (\ref{ode4}).  Without loss of generality, we suppose that the virus strains are ordered such that their corresponding reproduction numbers are in non-increasing order.  Given the motivation described in Section \ref{NestedIntro} we also assume that the immune responses are ordered with non-increasing dominance.  In addition, assume that the viral reproduction numbers are distinct from each other, and likewise for the immune reproduction numbers.  Thus the following holds:
 \begin{align}
 \mathcal R_1&>\mathcal R_2>\dots > \mathcal R_n \ \ \text{and} \ \   \mathcal I_1>\mathcal I_2>\dots > \mathcal I_n   \label{conditions}
 \end{align}
 
 Similar to the analysis contained in \cite{jover2013mechanisms,korytowski2015nested}, for $k \in \left\{1,2,\dots, n\right\}$ define the following quantities:
 \begin{align*}
\mathcal Q_{k}&:= 1+ \sum_{i=1}^{k} \mathcal R_i \left( s_i -s_{i-1} \right), \quad \text{where} \ \ s_i=\frac{1}{\mathcal I_i}, \ s_0=0  \\
&= 1+ \sum_{i=1}^{k - 1} s_i \left( \mathcal R_i -\mathcal R_{i+1} \right) + s_{k}\mathcal R_{k} \\
&=\mathcal Q_{k-1}+ (s_{k}-s_{k -1})\mathcal R_{k}, \quad \text{where} \ \ \mathcal Q_0=1.  
\end{align*} 
 
 For system (\ref{ode4}) with condition (\ref{conditions}), there are a multitude of non-negative equilibria ($>2^n$), but we only need to consider $2n+1$ of these equilibria for our global analysis.   First, define the infection-free equilibrium $\mathcal E_0  = (x^*,y^*,z^*)$, where $x^*=1$ and $y^*_i=z^*_i=0$ for $i=1,\dots,n$.  Then, for each $k\in\left\{1,\dots,n\right\}$, define the following equilibria:
 \begin{align}
 \mathcal E_k   = (x^*,y^*,z^*),  \qquad  & x^*=\frac{1}{\mathcal R_k}, \ y^*_i=s_i-s_{i-1}, \ z^*_i=\frac{\mathcal R_i-\mathcal R_{i+1}}{\mathcal R_k} \ \  \ \text{for} \ \ 1\leq i<k, \label{equilib1} \\  & y^*_k=1-\frac{ \mathcal Q_{k-1}}{\mathcal R_k}, \ z^*_k=0, \quad y^*_i=z^*_i=0 \  \ \ \text{for} \ \ k<i\leq n \notag \\  \notag \\
 \bar{\mathcal E}_k = (x^*,y^*,z^*),  \qquad &  x^*=\frac{1}{\mathcal Q_k}, \ y^*_i=s_i-s_{i-1}, \ z^*_i=\frac{\mathcal R_i-\mathcal R_{i+1}}{\mathcal Q_k} \ \  \ \text{for} \ \ 1\leq i<k,  \label{equilib2}  \\ &  y^*_k=s_k-s_{k-1}, \ z^*_k=\frac{\mathcal R_k}{\mathcal Q_k}-1, \quad  y^*_i=z^*_i=0 \  \ \ \text{for} \ \ k<i\leq n  \notag
 \end{align}
 Note that in the single virus (no immune) equilibrium $\mathcal E_1$, $y_1^*=1-\frac{1}{\mathcal R_1}$, $z_1^*=0$, and $y^*_i=z^*_i=0$ for all $1<i\leq n$.  Also notice that the equilibrium $\mathcal E_k$ contains $k$ positive and $n-k$ zero virus components, along with $k-1$ positive and $n-k+1$ zero immune response components.  The equilibrium $\bar{\mathcal E}_k$ contains $k$ positive and $n-k$ zero components for both virus and immune response.  
 
 In what follows, we will be interested the global behavior of solutions to system (\ref{ode4}).  In doing so, we will determine which viral strains and immune responses uniformly persist \cite{thieme1993persistence} and which go extinct.  Define a solution component, for example $y_i$, to be \emph{uniformly persistent} if 
 $$ \exists \ \epsilon>0 \ \text{(independent of positive initial conditions) such that} \  \liminf_{t\rightarrow\infty} y_i(t) >\epsilon. $$  We now state the main theorem characterizing the dynamics of the model.  
 \begin{theorem}\label{mainThm}
 Consider the model with perfectly nested network (\ref{ode4}) under the assumption of decreasing reproduction numbers (\ref{conditions}) and suppose positive initial conditions, i.e. $x(0),y_i(0),z_i(0)>0$ for all $i \in [1,n]$.  Then the stability of equilibria (\ref{equilib1}) and (\ref{equilib2}), and solution dynamics are characterized as follows:
  \begin{enumerate}
 \item If $\mathcal R_1 \leq 1$, then the infection-free equilibrium $\mathcal E_0$ is globally asymptotically stable.
 \item If $\mathcal R_1 > 1$ and $\mathcal R_1 \leq \mathcal Q_1$, then $\mathcal E_1$ is globally asymptotically stable.
 \item If $\mathcal R_1 > \mathcal Q_1$, let $k$ be the largest integer in $[1,n]$ such that $\mathcal R_k > \mathcal Q_k$.  
     \begin{itemize}
 \item[(a)]  If $k=n$ or $\mathcal R_{k+1} \leq \mathcal Q_{k}$, then $\bar{\mathcal E}_{k}$ is locally stable and the components $y_i,z_i$, $1\leq i\leq k$ are uniformly persistent.  Also, $\lim_{t\rightarrow\infty} x(t)=x^*$, and $\lim_{t\rightarrow\infty} y_i(t)=\lim_{t\rightarrow\infty} z_{i}(t)=0$ for $i\geq k+1$.  In addition, for the case $k=1,2$, $\bar{\mathcal E}_k$ is globally asymptotically stable.
 \item[(b)] If $k<n$ and $\mathcal R_{k+1} > \mathcal Q_{k} $, then $\mathcal E_{k+1}$ is locally stable and the components $y_i,z_i$, $1\leq i\leq k$ are uniformly persistent, and $\lim_{t\rightarrow\infty} y_{k+1}(t)=y_{k+1}^*$.  Also, $\lim_{t\rightarrow\infty} x(t)=x^*$, and $\lim_{t\rightarrow\infty} y_i(t)=0$ for $i\geq k+2$ and $\lim_{t\rightarrow\infty} z_{i}(t)=0$ for $i\geq k+1$.  In addition, for the case $k=1,2$, $\mathcal E_{k+1}$ is globally asymptotically stable.
 \end{itemize}
 Additionally, in each case, the long term averages of each component converge to the appropriate equilibrium value:  $$ \lim_{t\rightarrow\infty} \frac{1}{t} \int\limits_0^t y_i(t) \, dt = y_i^*, \quad \lim_{t\rightarrow\infty} \frac{1}{t} \int\limits_0^t z_i(t) \, dt = z_i^* . $$
 Moreover, in either case 3(a) or 3(b), the global attractor consists of positive ( and uniformly persistent, bounded) entire trajectories of System (\ref{ode4}) satisfying the following equations:
 \begin{align}
\dot y_i &= \gamma y_i\left(\sum_{j= i}^k \left( z_j^*-z_j \right)\right), \quad i=1,\dots,k \notag \\
  \dot z_i &= \frac{\sigma_i}{s_i} z_i\left(\sum_{j\leq i} \left( y_j - y_j^*\right) \right)  \label{invariantODE}  \\
  &\sum_{i=1}^k \mathcal R_i y_i =  \sum_{i=1}^k \mathcal R_i y_i^* \notag
  \end{align}
 \end{enumerate}
  \end{theorem}
 
 Theorem \ref{mainThm} describes the dynamics of system (\ref{ode4}) by showing local stability, uniform persistence and constraining the global asymptotic behavior of solutions.  The theorem generalizes uniform persistence results by Korytowski and Smith in \cite{korytowski2015nested} by allowing for mortality in the $y_i$ and $z_i$ populations (bacteria and phage populations in their context) and further characterizing the asymptotic dynamics.  Also, our method of proof is different.  In particular, we find a Lyapunov function which allows for the more complete dynamical analysis similar to the approaches in \cite{korytowski2,wolkowicz1989successful}.  While we are not global stability of the appropriate equilibria, $\bar{\mathcal E}_k$ or $\mathcal E_{k+1}$, when $k>2$ (when $y_1,y_2,z_1,z_2$ are uniformly persistent), the constrained dynamics of the global attractor make it likely that the equilibrium is globally stable.  
 
The dynamics on the global attractor can be further characterized.  Indeed since $\sum_{i=1}^k \mathcal R_i y_i =  \sum_{i=1}^k \mathcal R_i y_i^*$, we find that $0= \frac{1}{\gamma} \mathcal R_i \dot y_i $, and after further reduction:
$$z_k-z_k^*=\frac{1}{\mathcal Q_k -1}\sum_{j=1}^{k-1} \left( z_j^*-z_j \right) \sum_{i=1}^j  \mathcal R_i y_i  $$
Plugging this into the equations for $\dot y_i$ in (\ref{invariantODE}), the system reduces to $2(k-1)$ differential equations:
\begin{align}
\dot y_i &= \gamma y_i\left(\sum_{j= i}^{k-1} \left( z_j^*-z_j \right)\right)-\frac{\gamma y_i}{\sum_{i=1}^k \mathcal R_i y_i^*}\sum_{j=1}^{k-1} \left( z_j^*-z_j \right) \sum_{\ell=1}^j  \mathcal R_{\ell} y_{\ell} , \quad i=1,\dots,k-1 \label{Reduced_Sys}  \\
  \dot z_i &= \frac{\sigma_i}{s_i} z_i\left(\sum_{j\leq i} \left( y_j - y_j^*\right) \right),  \quad i=1,\dots,k-1 \notag
\end{align}
The global attractor has solutions which satisfy the above equations and the following constraints:
\begin{align*}
\sum_{i=1}^{k-1} \mathcal R_iy_i, \sum_{i=1}^{k-1} \mathcal R_iy_i < \mathcal Q_k -1, \epsilon< y_i(t),z_i(t)< M \ \ \forall t\in\mathbb R, 1\leq i\leq k-1,
\end{align*}  
where $\epsilon, M >0$ are uniform bounds.  The reduction of the invariant system to (\ref{Reduced_Sys}) allows us to prove global stability in the case $k=2$ with a similar approach to Lemma A.1. in \cite{wolkowicz1989successful}, where the analogous result for one-to-one interaction networks is established.   It remains an open question whether the global attractor can be proved to be the appropriate equilibrium for $k>2$.

\section{Application to Within-Host HIV/CTL Evolution} \label{SecApp}

Several studies have shown that the rate of HIV viral escape from CTL responses slows down after acute infection \cite{asquith2006inefficient,Vitaly1}.   Utilizing a multi-epitope HIV-CTL mathematical model, van Deutekom et al. \cite{vanDeutekom} argue that the contribution to killing of each individual CTL clone and escape rate is inversely associated with the breadth of the response, leading to declining escape rate after initial mutations.  They analyze the equilibrium density under the simplifying assumptions of identical viral strains and identical immune responses.  In contrast with the mass-action immune activation rate in our model (\ref{ode4}), their model has a saturating term for the activation rate which allows for coexistence of immune responses even in the case of a single virus strain.  Here in our model, we can obtain similar conclusions about diminishing escape rates, with the stronger analytical results.  In particular, as shown in the next paragraph, both the escape rate and equilibrium density of the invading strain decrease as immune response breadth increases.    

The escape rate from a CTL response, i.e. the rate of increase for a mutant virus resistant to that CTL response, can be calculated for Model (\ref{ode4}).  Consider the case where there are $k$ viral strains, $y_i$, and $k$ immune responses, $z_i$, ordered in decreasing fitness and dominance, respectively, as given in Conditions (\ref{conditions}).  Here, the perfectly nested structure implies $y_i$ is resistant to the immune responses $z_j$ where $j\leq i$.  If a mutant strain, $y_{k+1}$, arises which is resistant to $z_k$ (and all other immune responses $z_1,\dots,z_{k-1}$), then it can invade if $\mathcal R_{k+1}>\mathcal Q_{k}$, and converge to the component $y_{k+1}^*$ of equilibrium $\mathcal E_{k+1}$ by Theorem \ref{mainThm}.  The rate of invasion (assuming the system is at equilibrium $\bar{\mathcal E}_k$) is given by the  eigenvalue $\lambda_k$ of the linearized $\dot{y}_{k+1}$ equation:
\begin{align*}
\dot{y}_{k+1}=\gamma y_{k+1}\left( \frac{\mathcal R_{k+1}}{\mathcal Q_{k}} -1 \right)
\end{align*}
This invasion rate depends both upon the relative fitness of strain $y_{k+1}$ and the total selection pressure being applied by the immune response which depends on the breadth.  Indeed, 
\begin{align*}
\lambda_{k}=\gamma\left( \frac{\mathcal R_{k+1}}{\mathcal Q_{k}} -1 \right) < \gamma\left( \frac{\mathcal R_{k}}{\mathcal Q_{k-1}} -1 \right) =\lambda_{k-1},
\end{align*}
since $\mathcal R_{k+1}<\mathcal R_k$ and $\mathcal Q_{k+1}>\mathcal Q_k$.  Thus the escape rate from immune response $z_k$, denoted by $\lambda_k$, is less than $\lambda_{k-1}$.  In other words, the escape rate $\lambda_k$ is decreasing with respect to the current immune response breadth, $k$.  Similarly, the equilibrium density $y^*_{k+1}$ in $\mathcal E_{k+1}$ is less than $y^*_k$ in $\mathcal E_k$.   Additionally, since $\mathcal R_k$ decreases and $\mathcal Q_k$ increases with breadth $k$, Theorem \ref{mainThm} implies exclusion of $y_{k+1}$ is more likely as the breadth increases.  Therefore, the nested structure, where multi-epitope resistance is built up through sequential mutations, is consistent with observations of a relatively small number of escapes and diminishing escape rate.  

The invasion rate and equilibrium density of a subdominant immune response also diminishes with breadth of immune response and viral strain diversity.  Indeed, the invasion rate of immune response $z_{k+1}$ can be shown to be
\begin{align*}
\rho_{k+1}=\sigma_{k+1}\mathcal I_{k+1}\left(1- \frac{\mathcal Q_{k+1}}{\mathcal R_{k+1}}  \right).
\end{align*}
Assuming identical death rates for immune responses, i.e. $\sigma_i=\sigma$, we obtain that $\rho_{k+1}<\rho_k$.  Similarly, equilibrium density $z^*_{k+1}$ decreases and exclusion becomes more likely with increasing breadth.  Thus it may be difficult for a broadly acting (targeting conserved epitope) subdominant immune response to become established.

 \begin{figure}[t!]
\subfigure[][]{\label{fig1a}\includegraphics[width=7.5cm,height=4cm]{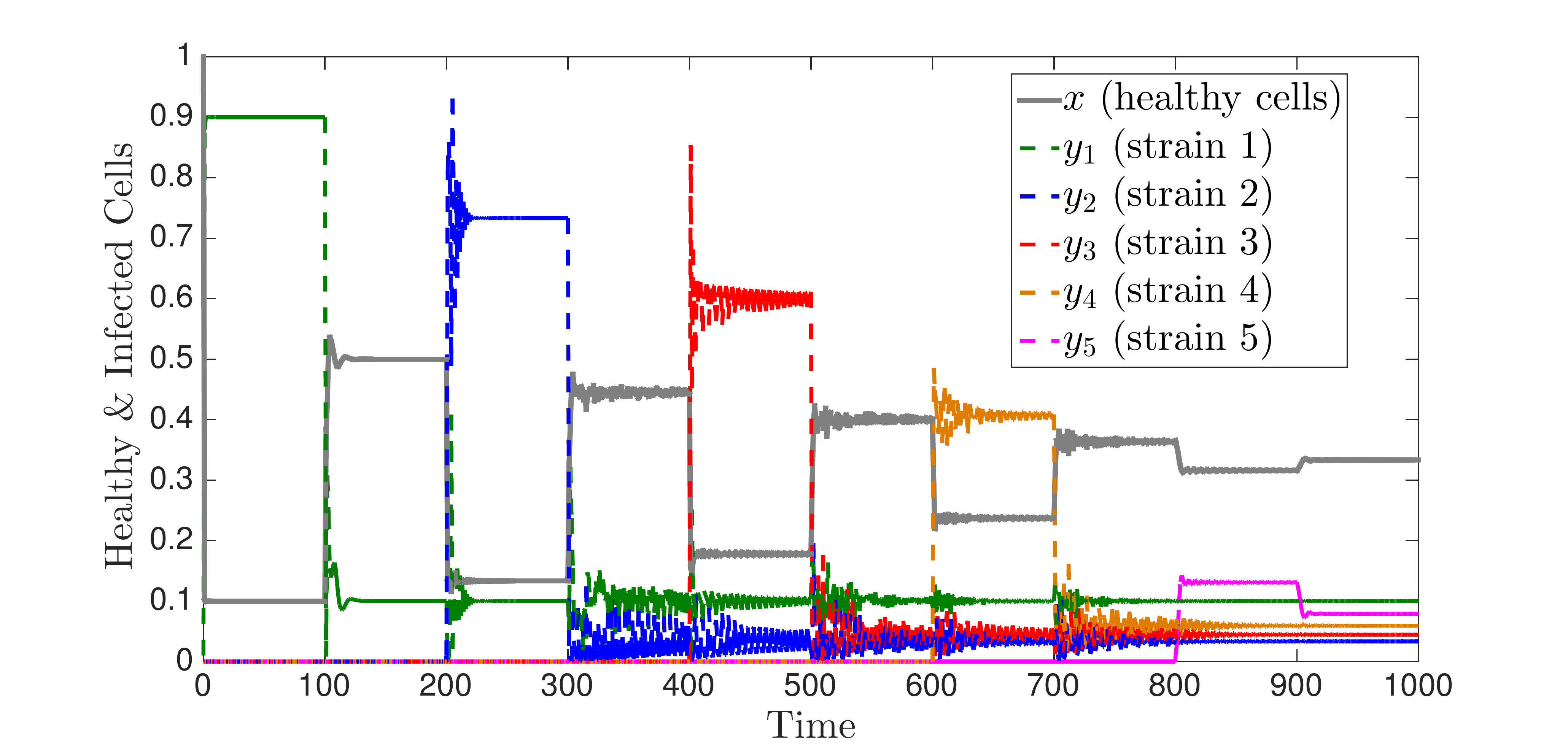}}
\subfigure[][]{\label{fig1b}\includegraphics[width=7.5cm,height=4cm]{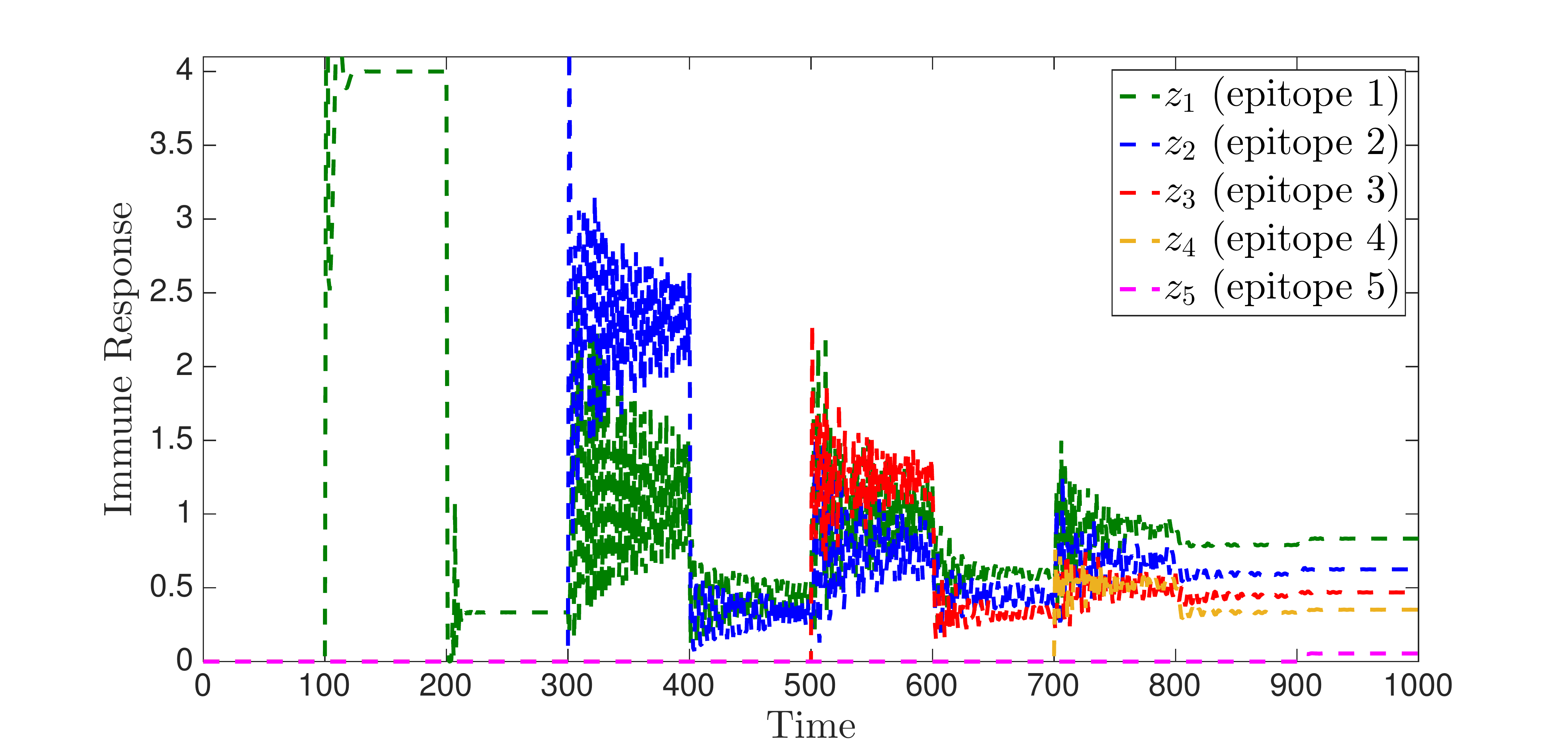}} \\
\subfigure[][]{\label{fig1a}\includegraphics[width=7.5cm,height=4cm]{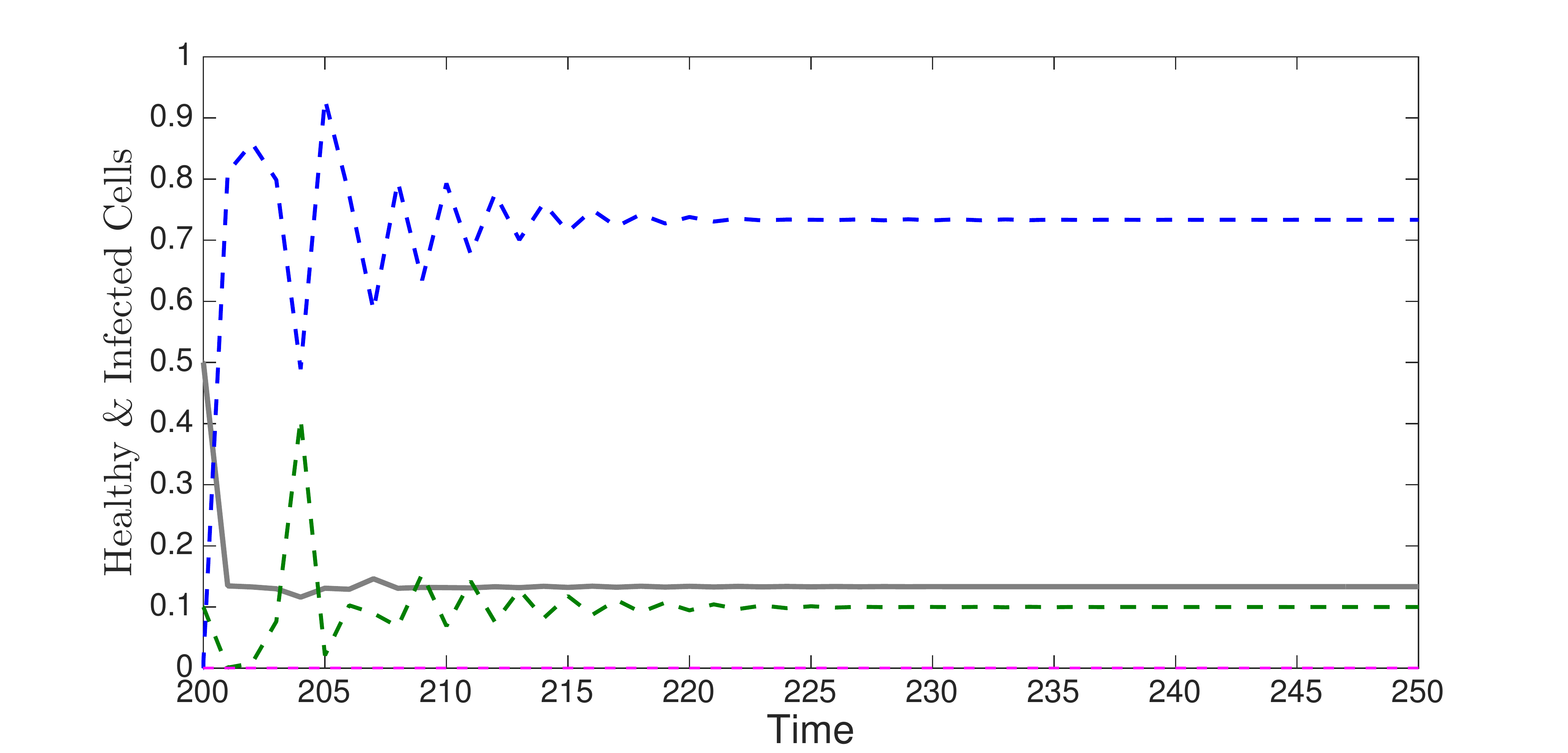}}
\subfigure[][]{\label{fig1b}\includegraphics[width=7.5cm,height=4cm]{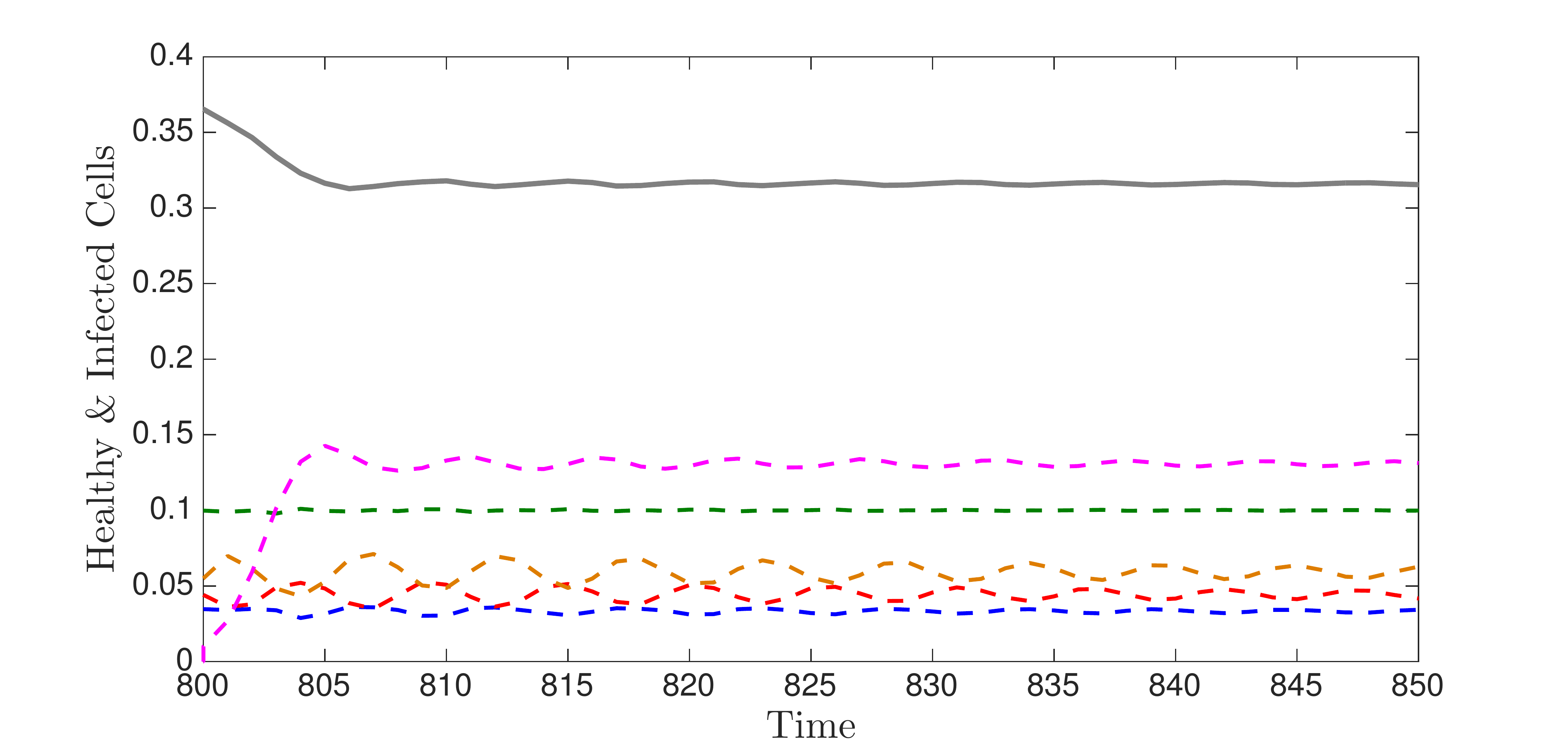}}
\caption{ \emph{Simulation of model (\ref{ode4}) with sequential introduction of mutant virus ``immune escape'' strains and subdominant immune responses}.  Initially at time 0, the wild-type virus (infected cells), $y_1$ is introduced into the healthy cell population.  Then, every 100 days, we alternate introductions of subdominant immune responses and escape virus strains, until $y_1,\dots, y_5, z_1,\dots, z_5$ persist after time 900.  Parameters values are $\mathcal R_1=10$, $\mathcal R_i=0.75 \mathcal R_{i-1}$ for $i=2,\dots,5$, $s_1=0.1$, $s_i=1.33 s_{i-1}$ for $i=2,\dots,5$, $\gamma=2.33$, and $\sigma_i=1.67$.  The healthy and infected cell populations are shown in (a), and immune responses in (b).  Notice the convergence to equilibria and the ``nested pattern'' of equilibrium healthy cells, which agree with analytical results.  In (c) and (d), we zoom in on the invasion by $y_2$ and $y_5$, respectively.   Consistent with linearized analysis and observed HIV/CTL patterns, the rate of invasion slows down, i.e. diminishes with immune and viral variant breadth. }
  \label{Fig1}
  \end{figure}

The total amount of healthy cells, infected cells and immune response at equilibrium density follows a ``nesting pattern'' with increasing breadth.  Let $\mathcal X_k^*$ and $\mathcal Y_k^*=\sum\limits_{i} y_i^*$ be the healthy cells total infected cells at equilibrium $\mathcal E_k$, respectively.  Similarly define $\bar{\mathcal X}_k^*$ and $\bar{\mathcal Y}_k^*$ at equilibrium $\bar{\mathcal E}_k$, along with $\mathcal Z_k^*$ and $\bar{\mathcal Z}_k^*$ as total immune response at equilibrium $\mathcal E_k$ and $\bar{\mathcal E}_k$ respectively.  Then the following relations hold:
\begin{align}
\mathcal X_1^*<\dots <\mathcal X_k^* & < \mathcal X^*_{k+1} < \bar{\mathcal X}^*_{k+1} < \bar{\mathcal X}^*_k< \dots <\bar{\mathcal X}^*_1, \qquad  \mathcal Y_1^*>\dots>\mathcal Y_k^*   > \bar{\mathcal Y}^*_k>\dots>\bar{\mathcal Y}^*_1, \notag \\  & \qquad \qquad 0=\mathcal Z_1^*<\dots <\mathcal Z_k^*  < \mathcal Z^*_{k+1} < \bar{\mathcal Z}^*_{k+1} < \bar{\mathcal Z}^*_k< \dots <\bar{\mathcal Z}^*_1.  \label{totEq}
\end{align}
 Notice that as breadth $k$ increases, total infected cell equilibrium is never as high as when there is wild type $y_1$ and no immune response (at $\mathcal E_1$), and never as low as when the first (immunodominant) immune response $z_1$ kicks in (at $\bar{\mathcal E}_1$).  Suppose the system is at equilibrium $\bar{\mathcal E}_k$, i.e. viral strain and immune response breadth $k$, and the escape mutant $y_{k+1}$ arises.  Convergence to $\mathcal E_{k+1}$ will follow (by Theorem \ref{mainThm}), increasing the total infected cells to $\mathcal Y^*_{k+1}$.  The infected cell count can be reduced to $\bar{\mathcal Y}_{k+1}^*$ by eliciting a subdominant immune response $z_{k+1}$ targeting a conserved epitope on $y_{k+1}$ (and all strains $y_1,\dots,y_k$).  While establishing this immune response $z_{k+1}$ certainly improves the health of the individual from the state at equilibrium $\mathcal E_{k+1}$ by targeting the escape variant $y_{k+1}$, the infected cell count has increased as the viral strain and immune breadth increased from $k$ to $k+1$ ($\bar{\mathcal Y}^*_{k+1} > \bar{\mathcal Y}^*_k$).  We perform numerical simulations of the model under the described scenario where sequential viral escape is followed by introduction of subdominant immune response targeting the escape variant, shown in Figure \ref{Fig1}.  

The above analysis reinforces the importance of strong immune responses directed at conserved epitopes (high fitness cost for resistance) in order to control HIV with CTL response.  If dominant immune responses can be readily escaped and broad responses attacking conserved epitopes are weak, then the system will evolve to have high infected cell load even when the overall immune breadth is large.  This is not to say that the dominant immune response is without benefit, since resistance will come at a fitness cost for the virus, as long as there are not compensatory mutations.  This is true in our model even though it is possible for certain parameters that a dominant response which induces a resistant strain, suppresses a broadly acting subdominant response to zero.  If the dominant response is removed from the system, the subdominant one would be able to persist and would not induce a resistant strain.  However, even in this case, the number of healthy cells is greater in the presence of the dominant strain (we omit analytical and numerical results pertaining to such a case).  Overall, our findings support the notion that an ideal CTL boosting vaccine or therapy for HIV would be able to target conserved epitopes with a strong response, which has been the conclusion of other modeling studies \cite{Althaus,liu2013vertical,vanDeutekom}.

The global analysis and extensive numerical simulations suggest that the perfectly nested structure, along with the assumed trade-offs between fitness and resistance (strain reactivity) for virus (immune response) given by Conditions (\ref{conditions}), produce stable solution dynamics where total populations converge to the ``nested'' steady states (\ref{totEq}).  Thus, system (\ref{ode4}) provides a model in which a diverse collection of viral strains and immune response can evolve and coexist at equilibrium stably with the nested immunodominance and resistance/fitness hierarchy.  The stable diversity of the virus strains (viral ``quasispecies'') and of the immune response (immune breadth) is made possible by the structure of the system.  In the absence of the immune response, the most fit virus competitively excludes all others \cite{browne2,de2008multistrain}, implying a form of predator mediated coexistence \cite{jmb-97}, and conversely if there is just one virus strain, then the dominant immune response will drive all other responses to extinction \cite{browne2016}.  Additionally, the nested network generates the stable coexistence while assuming epitope-specific immune responses with multi-strain reactivity, in contrast to the strain-specific nature of a one-to-one reactivity network.  

It is not clear how our results may be affected by altering key assumptions of the model.  Analysis is significantly more complicated in the case of general interaction networks which do not have the perfectly nested symmetry or one-to-one reactivity.  While analytical work is limited in the more general setting, Jover et al. simulated stable coexistence in their Lotka-Volterra model for infection networks which were not perfectly nested, but noted that the appropriate conditions for coexistence are difficult to find \cite{jover2013mechanisms}.   In a follow-up study to this work, we will address more general interaction networks.  In addition, our model neglects stochasticity which may cause certain persistent viral strains to become extinct. There have been studies of HIV-CTL dynamics which consider all possible mutational pathways from multiple epitopes in stochastic models \cite{Althaus,leviyang2015broad,Vitaly2,vanDeutekom,batorsky2014route}.  In one such study, van Deutekom et al. produced simulations where the persistent viral strains had a nested structure with respect to resistance \cite{vanDeutekom}.  Other work has emphasized the role of sequential epitope escapes, immunodominance and fitness costs in determining viral evolution \cite{da2012dynamics,kessinger2015inferring,liu2013vertical}.  All of this research suggests that nestedness and trade-offs between fitness and resistance/reactivity may evolve in HIV-CTL networks.  In this paper we determined which HIV/CTL populations persist and characterized stability of equilibria in an ideal case of perfectly nested networks.  Further modeling of the complex multi-epitope HIV-CTL dynamics may help to gain insight on immunotherapy and vaccine design, along with contributing to theory in community ecology and evolution.

\section{Proofs of theorems} \label{SecProofs}
\begin{proof}[Proof of Proposition \ref{bounded}:]
First, to establish non-negativity of solutions, note that the set where $Y_i=0$ or $Z_j=0$ is invariant.  In addition, the set where $X>0$ is invariant.  Non-negativity of solutions follows from these properties and the smoothness of the vector field.  Now define the quantity $S=X+\sum_{i=1}^m Y_i + \sum_{j=1}^n Z_j/q_j$.  It is not hard to see that $\frac{dS}{dt}=b-cS$ where $c=\min_{i,j} \left( a,\delta_i,\mu_j \right)$.  Thus $\limsup_{t\rightarrow\infty} S(t) \leq \frac{b}{c}$.  
\end{proof} 

\begin{proof}[Proof of Theorem \ref{mainThm}]
Define the following candidate Lyapunov function:
\begin{align*}
W(x,y,z)&=x-x^*\ln\frac{x}{x^*} + \sum_{i=1}^n \frac{1}{\gamma}\left(y_i-y_i^*\ln\frac{y_i}{y_i^*} \right) + \sum_{i=1}^n \frac{s_i}{\sigma_i}\left(z_i-z_i^*\ln\frac{z_i}{z_i^*} \right) \\
&:=W_1+W_2+W_3,
\end{align*}
where the term with logarithm should be omitted if the corresponding coordinate
is zero.  Then taking the time derivatives, we obtain
\begin{align*}
\dot W_1 &= 1- x -\frac{x^*}{x}+x^* + x^* + \sum_{i=1}^n \mathcal R_i y_i (x^*-x), \\
\dot W_2 &=  \sum_{i=1}^n \left(\mathcal R_i x-1-\sum_{j\geq i} z_j \right) (y_i-y_i^*), \\
\dot W_3 &=  \sum_{i=1}^n \left(\sum_{j\leq i} y_j - s_i \right) (z_i-z_i^*)
\end{align*}

Thus
\begin{align}
\dot W &= 1- x -\frac{x^*}{x}+x^* + \sum_{i=1}^n\left[ \mathcal R_i( y_i x^*-y_i^*x)-y_i+y_i^*\left(1+\sum_{j\geq i} z_j \right) -  z_i^*\sum_{j\leq i} y_j - s_i (z_i-z_i^*) \right]  \label{Wdot}
\end{align}

\textbf{Case 1}
Inserting the equilibrium components corresponding to $\mathcal E_0$ into (\ref{Wdot}), we obtain:
\begin{align*}
\dot W &= 2- x -\frac{1}{x} + \sum_{i=1}^n\left[ \left(\mathcal R_i-1\right)y_i - s_i z_i \right] \\
&\leq  \frac{-1}{x} \left( x-1\right)^2,
\end{align*}
since $\mathcal R_i<1$ for all $i$.   Thus, applying La Salle's Invariance principle, the $\omega-$limit set corresponding to any non-negative solution of (\ref{ode4}) (with $x(0)>0$) is contained in the largest invariant set where $\dot W =0$.  In this case, it is easy to see that $\dot W=0$ if and only if $y_i=z_i=0$ for all $1\leq i\leq n$ and $x=1$, i.e. $(x,y,z)= \mathcal E_0$.  

Next, we consider Case 3(a).  Note that we will consider Case 2 together with Case 3(b).

\textbf{Case 3(a):}
Suppose that $\mathcal R_1>1$ and $\mathcal R_1 \leq \mathcal Q_1$.  Let $k\in[1,n]$ be the largest integer such that $\mathcal R_k > \mathcal Q_k$ and suppose that $k=n$ or $\mathcal R_{k+1} \leq \mathcal Q_{k}$.  We consider the equilibrium $\bar{\mathcal E}_k$ and define the following quantities: $\rho^k_j=\mathcal R_j, \ 1\leq j\leq k, \ \rho^k_{k+1}=\mathcal Q_k$.   In this way the $y, z$ equilibrium components of $\bar{\mathcal E}_k$ can be conveniently written as
\begin{align*}
y_j^*&= s_j-s_{j-1},   \qquad (s_0=0) \\
z_j^*&= \frac{\rho^k_j-\rho^k_{j+1}}{\mathcal Q_k},
\end{align*}
for $j=1,\dots,k$.  Inserting the equilibrium values into (\ref{Wdot}), we find
\begin{align}
\dot W &= 1- x -\frac{1}{\mathcal Q_kx}+\frac{1}{\mathcal Q_k}  \notag \\ 
& \qquad\qquad \qquad+ \sum_{i=1}^k\left[ \mathcal R_i\left( \frac{y_i}{\mathcal Q_k} -(s_i-s_{i-1})x\right)-y_i+s_i-s_{i-1}+s_i \left( \frac{\rho^k_i-\rho^k_{i+1}}{\mathcal Q_k}-z_i\right) \right] \notag \\
& \qquad \qquad \qquad+\sum_{i=1}^k\left[ (s_i-s_{i-1})\sum_{j\geq i} z_j  -  \frac{\rho^k_i-\rho^k_{i+1}}{\mathcal Q_k}\sum_{j\leq i} y_j  \right] + \sum_{i=k+1}^n \left[ y_i\left( \frac{\mathcal R_i}{\mathcal Q_k} -1\right) -s_iz_i \right] \notag \\ \notag \\
&\leq 1- x -\frac{1}{\mathcal Q_kx}+\frac{1}{\mathcal Q_k} -(\mathcal Q_k-1)x+s_k+\frac{\mathcal Q_k -1-s_k\mathcal Q_k}{\mathcal Q_k}+ \sum_{i=1}^k\left[ y_i\left( \frac{\mathcal R_i}{\mathcal Q_k} -1\right) -s_iz_i \right]\notag \\
& \qquad \qquad \qquad+\sum_{i=1}^k\sum_{j\leq i} \left[ (s_j-s_{j-1})z_i  - \left( \frac{\rho^k_i-\rho^k_{i+1}}{\mathcal Q_k} \right)y_j  \right]  \quad (\text{since } \mathcal R_i \leq \mathcal Q_k \ \forall  i \geq k+1) \notag\\ \notag \\
&= \frac{-1}{\mathcal Q_k x}\left(\mathcal Q_k x -1\right)^2 + \sum_{i=1}^k\left[ y_i\left( \frac{\mathcal R_i}{\mathcal Q_k} -1\right)-s_iz_i +z_i \sum_{j\leq i}(s_j-s_{j-1}) + y_i\right] \notag \\
& \qquad \qquad \qquad   -\sum_{i=1}^{k-1}\frac{(\mathcal R_i-\mathcal R_{i+1})}{\mathcal Q_k} \sum_{j\leq i} y_j - \frac{\mathcal R_k}{\mathcal Q_k} \sum_{j\leq k} y_j \notag \\ \notag \\
&= \frac{-1}{\mathcal Q_k x}\left(\mathcal Q_k x -1\right)^2 + \frac{1}{\mathcal Q_k}\sum_{i=1}^{k-1}\left[ \mathcal R_i y_i -(\mathcal R_i-\mathcal R_{i+1})\sum_{j\leq i} y_j \right] +\frac{1}{\mathcal Q_k}\left[ \mathcal R_k y_k -\mathcal R_k\sum_{j\leq k} y_j \right]  \label{induct0}  \\
&= \frac{-1}{\mathcal Q_k x}\left(\mathcal Q_k x -1\right)^2 + \frac{1}{\mathcal Q_k}\sum_{i=1}^{k-1}\left[ -\mathcal R_i \sum_{j\leq i-1} y_j+\mathcal R_{i+1}\sum_{j\leq i} y_j -\mathcal R_k y_i\right] \notag
\end{align}

We claim that the second term (the sum from $i=1$ to $k-1$) is actually zero for $k\geq 1$.  First notice that this is true for $k=1$ by observing that the second and third term in (\ref{induct0}) are zero for the case $k=1$.  Next,
\begin{align}
\sum_{i=1}^{k-1}\left[ -\mathcal R_i \sum_{j\leq i-1} y_j+\mathcal R_{i+1}\sum_{j\leq i} y_j -\mathcal R_k y_i\right]  
&= \sum_{i=1}^{k-1}\left[ (-\mathcal R_i+\mathcal R_{i+1})\sum_{j\leq i-1} y_j +(\mathcal R_{i+1}-\mathcal R_k) y_i \right]  \label{induct1} \\
&= \sum_{i=1}^{k-1}(-\mathcal R_i+\mathcal R_{i+1})\sum_{j\leq i-1} y_j +\sum_{i=1}^{k-2}(\mathcal R_{i+1}-\mathcal R_k) y_i  \notag \\
&= \sum_{i=1}^{k-2}\left[ (-\mathcal R_i+\mathcal R_{i+1})\sum_{j\leq i-1} y_j +(\mathcal R_{i+1}-\mathcal R_{k-1}) y_i \right]  \label{induct2} 
\end{align}

Comparing equations (\ref{induct1}) and (\ref{induct2}), it is not hard to see by induction that the quantity of interest is indeed zero for all $k\geq 1$.  
Thus,
\begin{align*}
\dot W &= \frac{-1}{\mathcal Q_k x}\left(\mathcal Q_k x -1\right)^2+\sum_{i=k+1}^n \left[ y_i\left( \frac{\mathcal R_i}{\mathcal Q_k} -1\right) -s_iz_i \right]\leq \frac{-1}{\mathcal Q_k x}\left(\mathcal Q_k x -1\right)^2\leq 0
\end{align*}
Applying La Salle's Invariance principle, the $\omega-$limit set corresponding to any solution of (\ref{ode4}) originating in $\mathcal Z_k$ is contained in the largest invariant set where $\dot W =0$, denoted by $L$.  In this case $\mathcal R_i < \mathcal Q_k$ for $i\geq k+2$ and $\mathcal R_{k+1}\leq \mathcal Q_k$, which implies $\dot W=0$ iff $x=x^*=\frac{1}{\mathcal Q_k}$, $z_i=0$ for $i\geq k+1$, $y_i=0$ for $i\geq k+1$ or $\mathcal R_{k+1}= \mathcal Q_k$ and $y_i=0$ for $i\geq k+2$.  If $\mathcal R_{k+1}= \mathcal Q_k$, then $\dot y_{k+1}=0$ and $y_{k+1}=\bar y_{k+1}$ where $\bar y_{k+1}$ is constant.

On the invariant set $L$ a solution satisfies the following equations:
\begin{align}
\dot y_i &= \gamma y_i\left(\sum_{j= i}^k \left( z_j^*-z_j \right)\right), \quad i=1,\dots,k \label{inv1} \\
  \dot z_i &= \frac{\sigma_i}{s_i} z_i\left(\sum_{j\leq i} \left( y_j - y_j^*\right) \right) \notag
\end{align}
where $y_i^*$ and $z_i^*$, $i=1,\dots,k$, are the positive components of the equilibrium $\bar{\mathcal E}_k$. In addition,  $z_i=0 \ \text{for} \ k+1 \leq i\leq n$, and $y_i=0 \ \text{for} \ k+2 \leq i\leq n$, and
  \begin{align}
  \sum_{i=1}^k \mathcal R_i y_i &=  \sum_{i=1}^k \mathcal R_i y_i^* - \mathcal R_{k+1}\bar y_{k+1} \label{Ri1} 
  \end{align}
     Since $\dot W\leq 0$ and $W \rightarrow \infty$ as $y_i$, $z_i$ goes to $0$ or $\infty$ for any $i\in [1,k]$, we find that for any solution there exists $p,P>0$ such that $p\leq y_i,z_i\leq P$ for $i=1,\dots k$.   Consider a solution in the invariant set $L$.  For such a solution, by integrating the $\dot y_i$ equations, we obtain the following for $i=1,\dots, k$:
$$\frac{1}{\gamma T}\ln\left(\frac{y_i(T)}{y_i(0)}\right) = \sum_{j= i}^k \left( z_j^*- \frac{1}{T} \int\limits_0^T z_j(t)\,dt\right) $$
Letting $T\rightarrow \infty$, the left-hand side tends to zero.  Then, letting $i=k$, we obtain that $\frac{1}{T} \int\limits_0^T z_k(t)\,dt=z_k^*$.  Successively solving equations $i=k-1,k-2,\dots,1$, we find that $$\lim_{T\rightarrow\infty}\frac{1}{T} \int\limits_0^T z_i(t)\,dt=z_i^*,$$ for all $i=1,\dots,k$.  Similarly, by integrating $\dot z_i$ equations, we obtain $$\lim_{T\rightarrow\infty}\frac{1}{T} \int\limits_0^T y_i(t)\,dt=y_i^*.$$   
  Since the $\omega$-limit set of any solution is contained in $L$, these relations hold for any solution originating in $\mathcal Z_k$.  This implies that for any solution:
 $$\limsup_{t\rightarrow \infty}y_i(t) \geq y_i^* \quad \text{and} \quad \limsup_{t\rightarrow \infty}z_i(t) \geq z_i^*, \quad i=1,\dots, k $$
Therefore, $y_i$ and $z_i$, $1\leq i\leq k$, are uniformly weakly persistent.  Theorem \ref{bounded} implies that the key hypotheses of Corollary 4.8 from \cite{smith2011dynamical}
are satisfied, and thus weak uniform persistence implies strong
uniform persistence.  Notice also that by taking the asymptotic average of (\ref{Ri1}) we obtain that $y_{k+1}=\bar y_{k+1}=0$.  Thus the dynamics of the global attractor satisfy (\ref{invariantODE}), along with satisfying $\epsilon< y_i(t),z_i(t)< M \ \ \forall t\in\mathbb R, 1\leq i\leq k$, where $\epsilon, M >0$ are uniform bounds.

For the case $k=1$, we can prove global stability of $\bar{\mathcal E}_1$.  In this case, on the invariant set $L$ we have $y_2=\dots=y_n=0$ and $z_2=\dots=z_n=0$.  This implies that $y_1=y_1^*$.  Thus $\dot y_1=0$, which implies $z_1=z_1^*$.  Therefore in this case $L$ consists solely of the equilibrium $\bar{\mathcal E}_1$.

For the case $k=2$, we further reduce the system by (\ref{Reduced_Sys}), so that $(y_1,z_1)$ satisfies the planar system:
\begin{align}
\dot y_1 &= \gamma y_1 \left( z_1^*-z_1\right)\left(1-\frac{\mathcal R_1 y_1}{\sum_{i=1}^2 \mathcal R_i y_i^*}\right),  \label{reODEn}  \\
  \dot z_1 &= \frac{\sigma_1}{s_1} z_1 \left( y_1 - y_1^*\right) \notag
\end{align}
Note that $y_i,z_i \ i=1,2$ are uniformly persistent, $y_i,z_i\rightarrow 0$ as $t\rightarrow\infty$ for $i>2$, and  $\bar{\mathcal E}_1$ is the unique equilibrium with $y_i,z_i>0 \Leftrightarrow i=1,2$.  This rules out the case that omega limit sets consist of a heteroclinic orbit.  Therefore by the Poincar\'e-Bendixson criteria and the direction of the vector field, if the omega-limit set orbit $(y_1,z_1)$ is not an equilibrium, then it is a periodic orbit around the equilibrium $(y_1^*,z_1^*)$.  Without loss of generality assume that $y_1(0)=y_1^*$.  By (\ref{invariantODE}), $y_2(0)=y_2^*$ and $\dot z_i(0)=0$.  Furthermore 
\begin{align*}
\mathcal R_1 \ddot y_1 &= -\mathcal R_2 \ddot y_2 \Rightarrow \gamma^2 y_1^*\left(\sum_{j= i}^2 \left( z_j^*-z_j(0) \right)\right)^2 =-\gamma^2 y_2^*\left( z_2^*-z_2(0)\right)^2
\end{align*} 
The only way the above equation can be satisfied is if $z_i(0)=z_i^*$.  Then $y_i(0)=y_i^*, \ z_i(0)=z_i^*$.  Thus $\bar{\mathcal E}_2$ is globally asymptotically stable.

\textbf{Case 2 and 3(b):}
Suppose that $\mathcal R_1>1$.  If $\mathcal R_1 > \mathcal Q_1$, then let $k=0$.  If $\mathcal R_1 \leq \mathcal Q_1$.   Let $k\in[1,n]$ be the largest integer such that $\mathcal R_k > \mathcal Q_k$.  In this case, suppose that $0 \leq k<n$ and $\mathcal R_{k+1} > \mathcal Q_{k} $ (if $k=0$ then recall $\mathcal Q_0:=1$), and consider equilibrium $\mathcal E_{k+1}$.   Inserting the equilibrium components into (\ref{Wdot}):
\begin{align*}
\dot W &= 1- x -\frac{1}{x\mathcal R_{k+1}}+\frac{1}{\mathcal R_{k+1}} \\
& \qquad \qquad + \sum_{i=1}^k\left[ \mathcal R_i\left( \frac{y_i}{\mathcal R_{k+1}} -(s_i-s_{i-1})x\right)-y_i+s_i-s_{i-1}+s_i \left( \frac{\mathcal R_i-\mathcal R_{i+1}}{\mathcal R_{k+1}}-z_i\right) \right] \\
& \qquad \qquad +\sum_{i=1}^k\left[ (s_i-s_{i-1})\sum_{j\geq i} z_j  -  \frac{\mathcal R_i-\mathcal R_{i+1}}{\mathcal R_{k+1}}\sum_{j\leq i} y_j  \right] +\mathcal R_{k+1} \left[ \frac{y_{k+1}}{\mathcal R_{k+1}} - \left(1-\frac{\mathcal Q_k}{\mathcal R_{k+1}} \right)x\right] \\ & \qquad \qquad   - y_{k+1} +\left(1-\frac{\mathcal Q_k}{\mathcal R_{k+1}} \right) \left(1+\sum_{j\geq k+1} z_j \right) -s_{k+1}z_{k+1}+ \sum_{i=k+2}^n \left[ y_i\left( \frac{\mathcal R_i}{\mathcal R_{k+1}} -1\right) -s_iz_i \right] \\ \\
&= 1- x -\frac{1}{x\mathcal R_{k+1}}+\frac{1}{\mathcal R_{k+1}} -(\mathcal Q_k-1)x+s_k+\frac{\mathcal Q_k -1-s_k\mathcal R_{k+1}}{\mathcal R_{k+1}} - (\mathcal R_{k+1} - \mathcal Q_k)x   \\
& \qquad \qquad  + \sum_{i=1}^k\left[ y_i\left( \frac{\mathcal R_i}{\mathcal R_{k+1}} -1\right) -s_iz_i \right]  +\sum_{i=1}^k\sum_{j\leq i} \left[ (s_j-s_{j-1})z_i  - \left( \frac{\mathcal R_i-\mathcal R_{i+1}}{\mathcal R_{k+1}} \right)y_j  \right] \\ & \qquad \qquad + 1-\frac{\mathcal Q_k}{\mathcal R_{k+1}}+ \sum_{i=k+1}^n \frac{z_i}{\mathcal R_{k+1}} \left( \mathcal R_{k+1}s_k+\mathcal R_{k+1} - \mathcal Q_k-\mathcal R_{k+1}s_i \right)+ \sum_{i=k+2}^n y_i\left( \frac{\mathcal R_i}{\mathcal R_{k+1}} -1\right)   
\end{align*}
Notice that $\mathcal R_i< \mathcal R_{k+1}$ for $i\geq k+2$.  Also, for $i\geq k+2$:
\begin{align*}
\mathcal R_{k+1} - \mathcal Q_k-\mathcal R_{k+1}s_i &= \mathcal R_{k+1} - \mathcal Q_{k+1}  - \mathcal R_{k+1}\left(s_i-s_{k+1}  \right)  \leq  0,
\end{align*}
since $s_i\geq s_{k+1}$ for $i\geq k+1$ and $k$ is chosen so that $\mathcal R_{k+1}\leq Q_{k+1}$.  Thus
\begin{align}
\dot W 
& \leq 1- x -\frac{1}{x\mathcal R_{k+1}}+\frac{1}{\mathcal R_{k+1}} -(\mathcal Q_k-1)x+\frac{\mathcal Q_k -1}{\mathcal R_{k+1}}+ \sum_{i=1}^k\left[ y_i\left( \frac{\mathcal R_i}{\mathcal R_{k+1}} -1\right) -s_iz_i \right]  \notag \\
& \qquad \qquad  +\sum_{i=1}^k\sum_{j\leq i} \left[ (s_j-s_{j-1})z_i  - \left( \frac{\mathcal R_i-\mathcal R_{i+1}}{\mathcal R_{k+1}} \right)y_j  \right]  - (\mathcal R_{k+1} - \mathcal Q_k)x + 1-\frac{\mathcal Q_k}{\mathcal R_{k+1}} \notag \\ \notag \\
&= \frac{-1}{\mathcal R_{k+1} x}\left(\mathcal R_{k+1} x -1\right)^2 \notag \\
& \qquad \qquad+ \sum_{i=1}^k\left[ y_i\left( \frac{\mathcal R_i}{\mathcal R_{k+1}} -1\right) -s_iz_i  + \sum_{j\leq i} \left( (s_j-s_{j-1})z_i  - \left( \frac{\mathcal R_i-\mathcal R_{i+1}}{\mathcal R_{k+1}} \right)y_j \right) \right]  \notag \\ \notag \\
&= \frac{-1}{\mathcal R_{k+1} x}\left(\mathcal R_{k+1} x -1\right)^2 + \sum_{i=1}^k\left[ -y_i+ \frac{\mathcal R_{i+1}}{\mathcal R_{k+1}} y_i - \sum_{j\leq i-1} \left( \frac{\mathcal R_i-\mathcal R_{i+1}}{\mathcal R_{k+1}} \right)y_j \right]   \label{inductb0}
\end{align}
Note that the second term in (\ref{inductb0}) is zero for $k=0$ and $k=1$.  We claim that the second term (the sum from $i=1$ to $k$) is actually zero for $2\leq k\leq n-1$.  Indeed
\begin{align}
\sum_{i=1}^k & \left[ \left( \frac{\mathcal R_{i+1}}{\mathcal R_{k+1}}-1\right) y_i - \sum_{j\leq i-1} \left( \frac{\mathcal R_i-\mathcal R_{i+1}}{\mathcal R_{k+1}} \right)y_j \right] \notag \\
& \qquad \qquad= \sum_{i=1}^{k-1}\left( \frac{\mathcal R_{i+1}-\mathcal R_{k+1}}{\mathcal R_{k+1}} \right)y_i - \sum_{i=1}^{k} \left( \frac{\mathcal R_i-\mathcal R_{i+1}}{\mathcal R_{k+1}} \right)\sum_{j\leq i-1}y_j   \label{inductb1} \\
&\qquad \qquad =\sum_{i=1}^{k-1}\left[ \frac{\mathcal R_{i+1}-\mathcal R_{k+1}}{\mathcal R_{k+1}}   -\left( \frac{\mathcal R_k-\mathcal R_{k+1}}{\mathcal R_{k+1}} \right)\right]y_i - \sum_{i=1}^{k-1} \left( \frac{\mathcal R_i-\mathcal R_{i+1}}{\mathcal R_{k+1}} \right)\sum_{j\leq i-1}y_j  \notag \\
& \qquad \qquad=\sum_{i=1}^{k-2} \left( \frac{\mathcal R_{i+1}-\mathcal R_k}{\mathcal R_{k+1}} \right)y_i - \sum_{i=1}^{k-1} \left( \frac{\mathcal R_i-\mathcal R_{i+1}}{\mathcal R_{k+1}} \right)\sum_{j\leq i-1}y_j   \label{inductb2}
\end{align}
Comparing equations (\ref{inductb1}) and (\ref{inductb2}), it is not hard to see by induction that the quantity of interest is indeed zero for all $n-1\geq k\geq 2$.  Thus,
\begin{align*}
\dot W &=\frac{-1}{\mathcal R_{k+1} x}\left(\mathcal R_{k+1} x -1\right)^2+ \sum_{i=k+1}^n \frac{z_i}{\mathcal R_{k+1}} \left( \mathcal R_{k+1} - \mathcal Q_k-\mathcal R_{k+1}(s_i-s_{k+1}) \right)+ \sum_{i=k+2}^n y_i\left( \frac{\mathcal R_i}{\mathcal R_{k+1}} -1\right) \\
& \leq  \frac{-1}{\mathcal R_{k+1} x}\left(\mathcal R_{k+1} x -1\right)^2\leq 0, \qquad  \text{for} \ \ k=0,1,\dots, n-1.
\end{align*}
In this case, $\dot W=0$ iff $x=x^*=\frac{1}{\mathcal R_{k+1} }$, $y_i=0$ for $i\geq k+2$ and $z_i=0$ for $i\geq k+1$.  This implies that $\dot y_{k+1}=0$ and, thus $y_{k+1}=\bar y_{k+1}$ where $\bar y_{k+1}$ is constant.

On the invariant set $L$ a solution satisfies the following equations:
\begin{align*}
\dot y_i &= \gamma y_i\left(\sum_{j= i}^k \left( z_j^*-z_j \right)\right), \quad i=1,\dots,k \label{inv2} \tag{18}\\
  \dot z_i &= \frac{\sigma_i}{s_i} z_i\left(\sum_{j\leq i} \left( y_j - y_j^*\right) \right)
\end{align*}
where $y_i^*$ and $z_i^*$, $i=1,\dots,k$, are the positive components of the equilibrium $\mathcal E_{k+1}$. In addition,  $z_i=0 \ \text{for} \ k+1 \leq i\leq n$, and $y_i=0 \ \text{for} \ k+2 \leq i\leq n$, and
  \begin{align*}
  \sum_{i=1}^k \mathcal R_i y_i &=  \mathcal R_{k+1}-1 - \mathcal R_{k+1}\bar y_{k+1} \label{Ri2} \tag{19}
  \end{align*}
Noticing that (\ref{inv2}) is the same system as (\ref{inv1}) with different equilibrium components, we conclude that for any solution, the following holds:
$$\lim_{T\rightarrow\infty}\frac{1}{T} \int\limits_0^T y_i(t)\,dt=y_i^*, \qquad \lim_{T\rightarrow\infty}\frac{1}{T} \int\limits_0^T z_i(t)\,dt=z_i^*,$$ for all $i=1,\dots,k$.  Also, on the invariant set $L$, we have that by taking the asymptotic average of (\ref{Ri2}), $y_{k+1}= y_{k+1}^*=1-\frac{\mathcal Q_k}{\mathcal R_{k+1}}$.  In addition, similar arguments show uniform persistence. 

Also, using the fact that $y_{k+1}= y_{k+1}^*$, we find that formula (\ref{Ri2}) can be expressed as:
 \begin{align*}
  \sum_{i=1}^k \mathcal R_i y_i &=  \sum_{i=1}^k \mathcal R_i y_i^*,
  \end{align*}
  where these particular equilibrium components of $\mathcal E_{k+1}$ are the same as $\bar{\mathcal E}_k$.  Thus, system (\ref{inv2}) on the invariant set $L$ is also of the same form as in the case $\bar{\mathcal E}_k$, and the global attractor consists of uniformly ``persistent'' and bounded solutions to (\ref{invariantODE}).

For the case 2, i.e. $\mathcal R_1 > 1$ and $\mathcal R_1 \leq \mathcal Q_1$, then we can show $\mathcal E_1$ is globally asymptotically stable.  Indeed, $\dot W=0$ iff $x=x^*=\frac{1}{\mathcal R_{1} }$, $y_1=y_1^*=1-\frac{1}{\mathcal R_1}$, and all other components are zero.  Therefore in this case $L$ consists solely of the equilibrium $\mathcal E_1$.  

For the case 3(b) with $k=1$, i.e. $\mathcal Q_1<\mathcal R_2\leq \mathcal Q_2$, then we prove $\mathcal E_2$ is globally asymptotically stable.  Here, $\dot W=0$ iff $x=x^*=\frac{1}{\mathcal R_{1} }$, $y_2=y_2^*=1-\frac{\mathcal Q_1}{\mathcal R_2}$, $y_i=0$ for $i\geq 3$, $z_i=0$ for $i\geq 2$, and $\sum_{i=1}^2 \mathcal R_iy_i=\sum_{i=1}^2 \mathcal R_iy_i^*$.  The last relation implies that $y_1=y_1^*$.  Then $\dot{y}_1=0$ implies $z_1=z_1^*$.  Thus $L$ consists solely of the equilibrium $\mathcal E_2$.  For the case $k=2$, since orbits on the invariant set $L$ satisfy the reduced system (\ref{reODEn}) and $\sum_{i=1}^2 \mathcal R_iy_i=\sum_{i=1}^2 \mathcal R_iy_i^*$, the argument in case 3(a) for global stability of $\bar{\mathcal E}_2$ also works to prove that $\mathcal E_3$ is globally asymptotically stable.  
\end{proof}

\bibliography{Nested_References}

\begin{thebibliography}{10}
\providecommand{\url}[1]{{#1}}
\providecommand{\urlprefix}{URL }
\expandafter\ifx\csname urlstyle\endcsname\relax
  \providecommand{\doi}[1]{DOI~\discretionary{}{}{}#1}\else
  \providecommand{\doi}{DOI~\discretionary{}{}{}\begingroup
  \urlstyle{rm}\Url}\fi

\bibitem{Althaus}
Althaus, C.L., Boer, R.D.: Dynamics of immune escape during hiv/siv infection.
\newblock PLoS Computational Biology \textbf{4}, e1000,103 (2008)

\bibitem{asquith2006inefficient}
Asquith, B., Edwards, C.T., Lipsitch, M., McLean, A.R.: Inefficient cytotoxic t
  lymphocyte--mediated killing of hiv-1--infected cells in vivo.
\newblock PLoS Biol \textbf{4}(4), e90 (2006)

\bibitem{batorsky2014route}
Batorsky, R., Sergeev, R.A., Rouzine, I.M.: The route of hiv escape from immune
  response targeting multiple sites is determined by the cost-benefit tradeoff
  of escape mutations.
\newblock PLoS Comput Biol \textbf{10}(10), e1003,878 (2014)

\bibitem{Blyuss}
Blyuss, K.B., Gupta, S.: Stability and bifurcations in a model of antigenic
  variation in malaria.
\newblock Journal of mathematical biology \textbf{58}(6), 923--937 (2009)

\bibitem{bobko2015singularly}
Bobko, N., Zubelli, J.P.: A singularly perturbed hiv model with treatment and
  antigenic variation.
\newblock Mathematical biosciences and engineering: MBE \textbf{12}(1), 1--21
  (2015)

\bibitem{browne2016}
Browne, C.: Immune response in virus model structured by cell infection-age.
\newblock Mathematical biosciences and engineering: MBE \textbf{13}(5) (2016)

\bibitem{browne2}
Browne, C.J.: A multi-strain virus model with infected cell age structure:
  Application to hiv.
\newblock Nonlinear Analysis: Real World Applications \textbf{22}, 354--372
  (2015)

\bibitem{PilyuginMalaria}
De~Leenheer, P., Pilyugin, S.S.: Immune response to a malaria infection:
  properties of a mathematical model.
\newblock Journal of biological dynamics \textbf{2}(2), 102--120 (2008)

\bibitem{de2008multistrain}
De~Leenheer, P., Pilyugin, S.S.: Multistrain virus dynamics with mutations: a
  global analysis.
\newblock Mathematical Medicine and Biology \textbf{25}(4), 285--322 (2008)

\bibitem{vanDeutekom}
Deutekom, H.V., Wijnker, G., Boer, R.D.: The rate of immune escape vanishes
  when multiple immune responses control an hiv infection.
\newblock Journal of immunology \textbf{191}, 3277--3286 (2013)

\bibitem{Vitaly1}
Ganusov, V.V., Goonetilleke, N., Liu, M.K., Ferrari, G., Shaw, G.M., Borrow,
  A.J.M.P., Korber, B.T., Perelson, A.S.: Fitness costs and diversity of the
  cytotoxic t lymphocyte (ctl) response determine the rate of ctl escape during
  acute and chronic phases of hiv infection.
\newblock Journal of virology \textbf{85}(20), 10,518--10,528 (2011)

\bibitem{Vitaly2}
Ganusov, V.V., Neher, R.A., Perelson, A.S.: Mathematical modeling of escape of
  hiv from cytotoxic t lymphocyte responses.
\newblock Journal of Statistical Mechanics: Theory and Experiment
  \textbf{2013.01}, P01,010 (2013)

\bibitem{hofbauer1998evolutionary}
Hofbauer, J., Sigmund, K.: Evolutionary games and population dynamics.
\newblock Cambridge university press (1998)

\bibitem{iwasa2004some}
Iwasa, Y., Michor, F., Nowak, M.: Some basic properties of immune selection.
\newblock Journal of theoretical biology \textbf{229}(2), 179--188 (2004)

\bibitem{jover2013mechanisms}
Jover, L.F., Cortez, M.H., Weitz, J.S.: Mechanisms of multi-strain coexistence
  in host--phage systems with nested infection networks.
\newblock Journal of theoretical biology \textbf{332}, 65--77 (2013)

\bibitem{kessinger2015inferring}
Kessinger, T.A., Perelson, A.S., Neher, R.A.: Inferring hiv escape rates from
  multi-locus genotype data.
\newblock Immune system modeling and analysis p. 348 (2015)

\bibitem{korytowski2015nested}
Korytowski, D.A., Smith, H.L.: How nested and monogamous infection networks in
  host-phage communities come to be.
\newblock Theoretical ecology \textbf{8}(1), 111--120 (2015)

\bibitem{korytowski2}
Korytowski, D.A., Smith, H.L.: Persistence in phage-bacteria communities with
  nested and one-to-one infection networks.
\newblock arXiv:1505.03827  (2015)

\bibitem{leviyang2015broad}
Leviyang, S., Ganusov, V.V.: Broad ctl response in early hiv infection drives
  multiple concurrent ctl escapes.
\newblock PLoS Comput Biol \textbf{11}(10), e1004,492 (2015)

\bibitem{liu2013vertical}
Liu, M.K., Hawkins, N., Ritchie, A.J., Ganusov, V.V., Whale, V., Brackenridge,
  S., Li, H., Pavlicek, J.W., Cai, F., Rose-Abrahams, M., et~al.: Vertical t
  cell immunodominance and epitope entropy determine hiv-1 escape.
\newblock The Journal of clinical investigation \textbf{123}(1), 380--393
  (2013)

\bibitem{nowak1996population}
Nowak, M.A., Bangham, C.R.: Population dynamics of immune responses to
  persistent viruses.
\newblock Science \textbf{272}(5258), 74--79 (1996)

\bibitem{Nowak2}
Nowak, M.A., May, R.M., Sigmund, K.: Immune responses against multiple
  epitopes.
\newblock Journal of theoretical biology \textbf{175}(3), 325--353 (1995)

\bibitem{pandit2014reliable}
Pandit, A., de~Boer, R.J.: Reliable reconstruction of hiv-1 whole genome
  haplotypes reveals clonal interference and genetic hitchhiking among immune
  escape variants.
\newblock Retrovirology \textbf{11}(1), 1 (2014)

\bibitem{Perelson2}
Perelson, A.S., Nelson, P.W.: Mathematical analysis of hiv-1 dynamics in vivo.
\newblock SIAM review \textbf{41}(1), 3--44 (1999)

\bibitem{Gupta}
Recker, M., Gupta, S.: Conflicting immune responses can prolong the length of
  infection in plasmodium falciparum malaria.
\newblock Bulletin of mathematical biology \textbf{68}(4), 821--835 (2006)

\bibitem{jmb-97}
Schreiber, S.J.: Gerneralist and specialist predators that mediate permanence
  in ecological communities.
\newblock Journal of Mathematical Biology \textbf{36}, 133--148 (1997)

\bibitem{da2012dynamics}
da~Silva, J.: The dynamics of hiv-1 adaptation in early infection.
\newblock Genetics \textbf{190}(3), 1087--1099 (2012)

\bibitem{smith2011dynamical}
Smith, H.L., Thieme, H.R., Thieme, H.R.: Dynamical systems and population
  persistence, vol. 118.
\newblock American Mathematical Society Providence, RI (2011)

\bibitem{souza2011global}
Souza, M.O., Zubelli, J.P.: Global stability for a class of virus models with
  cytotoxic t lymphocyte immune response and antigenic variation.
\newblock Bulletin of mathematical biology \textbf{73}(3), 609--625 (2011)

\bibitem{thieme1993persistence}
Thieme, H.R.: Persistence under relaxed point-dissipativity (with application
  to an endemic model).
\newblock SIAM Journal on Mathematical Analysis \textbf{24}(2), 407--435 (1993)

\bibitem{weitz2013phage}
Weitz, J.S., Poisot, T., Meyer, J.R., Flores, C.O., Valverde, S., Sullivan,
  M.B., Hochberg, M.E.: Phage--bacteria infection networks.
\newblock Trends in microbiology \textbf{21}(2), 82--91 (2013)

\bibitem{wolkowicz1989successful}
Wolkowicz, G.S.: Successful invasion of a food web in a chemostat.
\newblock Mathematical Biosciences \textbf{93}(2), 249--268 (1989)

\end{thebibliography}
\bibliographystyle{spmpsci}
\end{document}